%% file: main.tex
\documentclass[12pt]{article}

\usepackage{amsmath,amssymb,amsthm}

\usepackage{newtxtext}
\usepackage[nosymbolsc,amssymbols]{newtxmath}

\usepackage[margin=1in]{geometry}
\usepackage{setspace}
\raggedright
\usepackage{graphicx}
\usepackage{booktabs}
\usepackage{multirow}
\usepackage{array}

\usepackage{natbib}
\usepackage{hyperref}
\hypersetup{
  pdftitle={Design-Based Variance Estimation for Modern Heterogeneity-Robust Difference-in-Differences Estimators},
  pdfauthor={Isaac Gerber},
  pdfkeywords={difference-in-differences, survey statistics, design-based variance, influence functions, Binder linearization},
  colorlinks=true,
  linkcolor=blue,
  citecolor=blue,
  urlcolor=blue
}

\newtheorem{proposition}{Proposition}
\newtheorem{assumption}{Assumption}
\newtheorem{remark}{Remark}

\newcommand{\E}{\mathbb{E}}

\newcommand{\IF}{\mathrm{IF}}
\newcommand{\ATT}{\mathrm{ATT}}
\newcommand{\DEFF}{\mathrm{DEFF}}
\newcommand{\diag}{\mathrm{diag}}

\title{Design-Based Variance Estimation for Modern\\
  Heterogeneity-Robust Difference-in-Differences Estimators}
\author{Isaac Gerber\thanks{Meta Platforms, Inc.}}
\date{May 2026}

\begin{document}
\maketitle

\begin{abstract}
Modern heterogeneity-robust difference-in-differences estimators typically
derive their asymptotic properties under iid, cluster, or fixed-design
frameworks that abstract from complex survey sampling, yet practitioners
routinely apply them to nationally representative surveys with complex
stratified cluster designs. To our knowledge, prior work has not explicitly applied survey-sampling Taylor
linearization to the influence-function representations of these estimators.
We bridge this gap by showing that, under standard regularity conditions, the
influence functions established in the literature for each smooth IF-based or
regression-based modern DiD estimator considered here satisfy Binder's (1983)
smoothness conditions, so that applying the standard stratified-cluster
variance formula to their values produces design-consistent standard errors. A Monte Carlo study with 66,000
replications shows where the design effect comes from. HC1 standard errors
that treat observations as independent and unweighted produce coverage as
low as 34\% under a baseline survey design and below 11\% under
informative sampling. Combining the survey-weighted point estimate with
PSU-level clustering --- the practitioner's \texttt{cluster=psu} heuristic
--- recovers near-nominal coverage across all four scenarios, including
under informative sampling. Adding strata and finite-population
corrections (the full design specification) yields incremental precision
but is not required for valid coverage in our simulations. Our simulations further demonstrate that survey-weighted doubly robust
estimation with covariate adjustment produces well-calibrated inference when
parallel trends hold only conditionally. An empirical illustration using NHANES data on the ACA's
dependent coverage provision confirms that both point estimates and standard
errors change substantively --- enough to reverse significance conclusions ---
when the survey design is accounted for. To facilitate adoption, we provide
\texttt{diff-diff} (\url{https://github.com/igerber/diff-diff}), an
open-source Python package implementing design-based variance for fifteen
modern DiD estimators.
\end{abstract}


% ======================================================================
\section{Introduction}
\label{sec:introduction}

Policy evaluations frequently rely on nationally representative surveys ---
NHANES for health outcomes, ACS for demographics and housing,
CPS for labor force dynamics, and MEPS for medical expenditure. These surveys employ stratified multi-stage cluster sampling to
achieve national coverage at manageable cost. The resulting data carry two
features that invalidate iid-based standard errors: observations within the
same primary sampling unit (PSU) are correlated, and stratification constrains
sampling variability. The ratio of design-based variance to the iid HC1
variance --- the design effect (DEFF) --- commonly ranges from 2 to 5 in
health and social surveys.

This matters especially for difference-in-differences (DiD) estimation.
Treatment is often assigned at geographic levels --- state policies, county
programs, school-district mandates --- that encompass multiple survey PSUs,
so within-PSU homogeneity in treatment status intensifies the design effect. DiD estimands involve
contrasts across groups and time periods, amplifying distortions in variance
estimation. Incorrect standard errors can flip significance conclusions for
policy-relevant effect sizes, undermining the credibility of program
evaluations.

\subsection*{The gap in modern DiD theory}

The modern heterogeneity-robust DiD literature derives estimators and their
asymptotic properties under sampling assumptions that are incompatible with
complex survey designs. The foundational papers either assume iid sampling
explicitly or adopt frameworks that do not incorporate strata, PSU clustering,
or finite population corrections:

\begin{itemize}
\item \citet{callawaysantanna2021} state iid as a numbered assumption
  (Assumption~2) and derive the multiplier bootstrap under it.
\item \citet{santannazhao2020} assume iid (Assumption~1) and derive the
  doubly robust influence function and semiparametric efficiency bounds.
\item \citet{borusyak2024} adopt a conditional/fixed-design framework that
  avoids random sampling assumptions, conditioning on the observation set.
\item \citet{sunabraham2021} maintain iid as an unstated but operative
  assumption in deriving the interaction-weighted estimator.
\item \citet{dechaisemartin2020} assume group-level independence
  (Assumption~3), which does not map to stratified-cluster survey data.
\item \citet{gardner2022} invokes standard GMM regularity conditions that
  implicitly require iid or ergodic stationary data.
\end{itemize}

The most comprehensive recent synthesis --- \citet{roth2023} --- discusses
design-based inference in the treatment-assignment sense \citep{atheyimbens2022},
where randomness comes from which units receive treatment. It does not address
survey sampling design, survey weights, or strata/PSU/FPC-based variance
estimation.

A note on terminology is warranted. The recent DiD literature uses
``design-based'' to refer to treatment-assignment design
\citep{atheyimbens2022}, where uncertainty arises from which units receive
treatment. Throughout this paper, ``design-based'' refers to survey sampling
design \citep{binder1983}, where uncertainty arises from which units are
sampled. Same term, different referent.

\subsection*{The gap in software}

Existing software implementations reflect this theoretical gap. R's \texttt{did}
package \citep{callawaysantanna2021} accepts a \texttt{weightsname} parameter
for point estimation and supports cluster-level multiplier bootstrap, but does
not account for stratification or finite population corrections. Stata's
\texttt{csdid} accepts \texttt{pweight} for point estimation but does not
support the \texttt{svy:} prefix --- there is no mechanism for strata or FPC.
Other packages --- \texttt{did\_multiplegt\_dyn}, \texttt{eventstudyinteract},
\texttt{didimputation} --- likewise support weights for point estimation and
allow cluster-robust inference, but none provides full survey-design variance
estimation that jointly accounts for strata, PSU clustering, and finite
population corrections.

\subsection*{Adjacent work}

The survey statistics literature has developed design-based variance theory for
smooth functionals \citep{binder1983, demnati2004, lumley2004}, and recent work
has extended this to causal inference --- but primarily for cross-sectional
estimators or simple two-period designs, not for modern staggered DiD.

\citet{dugoff2014} provide practical guidance on combining propensity score
methods with complex surveys, but address cross-sectional treatment effects,
not DiD. \citet{zeng2025} derive sandwich variance for survey-weighted
propensity score estimators using influence functions --- the closest work to
the bridge we describe --- but for cross-sectional IPW/augmented weighting, not
staggered DiD. \citet{ye2025} study DiD with repeated cross-sectional survey
data, combining propensity scores with survey weights. However, their estimator
is limited to two periods and two groups, uses bootstrap-only variance (no
analytical design-based derivation), and does not address the modern
heterogeneity-robust estimators considered here.

To our knowledge, prior work has not explicitly applied survey-sampling Taylor
linearization --- in the strata/PSU/FPC sense --- to the influence-function
representations of modern heterogeneity-robust DiD estimators.

\subsection*{Contribution}

This paper bridges the two literatures. The core argument (Section~\ref{sec:theory})
is a verification step: we show that the influence-function representations
established in the original papers for each modern DiD estimator satisfy
Binder's smoothness conditions, so that the standard survey-statistics Taylor
linearization applies. Binder's theorem then implies that the stratified-cluster
variance formula applied to these influence-function values produces a
design-consistent variance estimator. The contribution is not the proposition
itself --- which is a direct corollary of \citet{binder1983} --- but the
explicit identification of these estimators as covered by it, and the
empirical demonstration that the resulting standard errors are well-calibrated
in finite samples.

We demonstrate the practical consequences through a Monte Carlo study
(Section~\ref{sec:simulation}) showing that HC1 confidence interval coverage
can fall as low as 34\% at the nominal 95\% level, with design effects of
2--17 times the HC1 variance in the baseline scenario and exceeding 100 times
under informative sampling. The simulation also decomposes where the
correction has to come from. Two ingredients matter for valid coverage: the
survey-weighted point estimate (which corrects estimand bias under
informative sampling) and PSU-level clustering (which corrects within-cluster
correlation). Combining the two gives near-nominal coverage in every
scenario we examine. Adding strata and finite-population corrections is real
but incremental --- a precision gain on top of an already-valid
specification, not a coverage requirement. An empirical illustration using NHANES
data on the ACA's dependent coverage provision
(Section~\ref{sec:empirical}) shows that ignoring the survey design changes
both the point estimate (by 48\%) and the significance conclusion.
Section~\ref{sec:software} describes the \texttt{diff-diff} Python package
(v3.3.2; \citealp{diffdiff2026}), which implements this connection and is,
to our knowledge, the first open-source implementation that jointly
supports strata, PSU clustering, finite population corrections, and
replicate-weight methods for these modern DiD estimators.
Section~\ref{sec:discussion} discusses limitations and directions for
future work.

% ======================================================================
\section{Setup and Notation}
\label{sec:notation}

\subsection{Finite population and survey design}

Consider a finite population $\mathcal{U} = \{1, \ldots, N\}$. The population
is partitioned into $H$ non-overlapping strata. Within stratum $h$, there are
$N_h$ PSUs in the population, of which $n_h$ are sampled. This describes the
standard stratified cluster design used by most federal statistical agencies.
The variance estimator in Proposition~\ref{prop:main} is a first-stage
formula that treats within-PSU totals as the unit of analysis; when
observations are sub-sampled within PSUs (as in multi-stage designs), this
is the standard ``ultimate cluster'' approximation \citep{lumley2004}, which
is exact when PSUs are fully enumerated. When observations are sub-sampled
within PSUs, including the first-stage FPC $(1 - f_h)$ without a second-stage
variance component can be slightly anti-conservative; omitting the FPC
(i.e., assuming with-replacement sampling at the first stage) restores the
conservative property.

Each sampled observation $i$ carries a raw sampling weight
$w_i^* = 1/\pi_i$, where $\pi_i$ is the inclusion probability. The raw
weights estimate the finite-population size: $\hat{N} = \sum_i w_i^*
= n \bar{w}^*$, where $\bar{w}^* = n^{-1} \sum_i w_i^*$. For computation,
we work with normalized weights $w_i = w_i^* / \bar{w}^*$, so that
$\sum_i w_i = n$. Throughout, $w_i$ denotes the normalized weight, and we
write $\hat{W} \equiv \sum_i w_i = n$ for the normalized-weight total
(distinct from the population estimate $\hat{N}$). Normalization preserves
relative representativeness and avoids scale dependence in regression
coefficients; the ratio $w_i / \hat{W} = w_i^* / \hat{N}$ is invariant to
normalization, so weighted averages take the same value either way.

The sampling fraction in stratum $h$ is $f_h = n_h / N_h$. When $f_h$ is close
to 1, most of the finite population has been observed and sampling variability
is reduced. The finite population correction factor $(1 - f_h)$ enters the
variance formula to account for this.

\subsection{Target estimand and source of uncertainty}

We formalize the estimand as $\theta = T(F)$, where $T$ is a functional mapping
a distribution to a real number (or vector) and $F$ is the finite-population
distribution of potential outcomes, treatment status, and covariates. For DiD,
$T$ extracts average treatment effects on the treated (ATTs) from this joint
distribution.

Throughout this paper, uncertainty arises exclusively from the survey sampling
mechanism: the population values $\{Y_i(0), Y_i(1), D_i, X_i\}$ are treated as
fixed, and randomness enters through which units are sampled. This is the
standard design-based framework of \citet{binder1983}, and it is the reason
that the variance estimator includes a finite population correction $(1 - f_h)$
that drives variance to zero as the sampling fraction approaches one. This
framework is distinct from the treatment-assignment uncertainty considered by
\citet{atheyimbens2022}, and from the superpopulation framework in which
parallel trends is typically stated. We adopt the convention that the parallel
trends assumption holds in the finite population; the design-based variance
then captures uncertainty about which members of that population we observe.
Survey weights affect the estimand (they determine which population the ATT
targets) as well as the variance (they determine how precisely we estimate it).

Under the survey design, the survey-weighted empirical distribution is the
H\'{a}jek (self-normalized) form:
\begin{equation}
  \hat{F}_w = \frac{\sum_i w_i \, \delta_{x_i}}{\sum_i w_i},
  \label{eq:hajek}
\end{equation}
where the sum is over sampled observations and $\delta_{x_i}$ is the point
mass at $x_i$. When $T$ is a smooth functional, the plug-in estimator
$\hat{\theta} = T(\hat{F}_w)$ is design-consistent for $\theta = T(F)$: as the
sample size grows within the finite-population asymptotic framework,
$\hat{\theta}$ converges in probability to $\theta$.

The abstraction of $\theta$ as a functional of $F$ is what enables the bridge
between survey statistics and DiD: both literatures reason about functionals,
from different perspectives.

\subsection{Survey-weighted estimation}

For regression-based estimators, the point estimates solve weighted least
squares (WLS), minimizing $\sum_i w_i (Y_i - X_i'\beta)^2$. For
influence-function-based estimators, point estimates are constructed from
survey-weighted sample moments, replacing simple sample averages $(1/n)\sum_i$
with weighted averages $(\sum_i w_i)^{-1} \sum_i w_i$. For doubly robust and
IPW variants, the same principle applies to propensity score estimation (via
survey-weighted logistic regression) and outcome regression.

Under the design-based perspective adopted here, survey weights are needed to
ensure that treatment effect estimates correspond to the finite population, not
just the sample \citep{solon2015}. Without weights, ATT estimates reflect the
sample composition, which may over-represent certain strata due to the sampling
design.

Table~\ref{tab:notation} summarizes the notation used throughout the paper.

\begin{table}[ht]
\centering
\caption{Notation summary.}
\label{tab:notation}
\begin{tabular}{@{}ll@{}}
\toprule
Symbol & Definition \\
\midrule
$\mathcal{U} = \{1, \ldots, N\}$ & Finite population \\
$H$ & Number of strata \\
$n_h$ & Sampled PSUs in stratum $h$ \\
$N_h$ & Total PSUs in stratum $h$ (for FPC) \\
$f_h = n_h / N_h$ & Sampling fraction in stratum $h$ \\
$w_i^* = 1/\pi_i$ & Raw sampling weight \\
$w_i = w_i^* / \bar{w}^*$ & Normalized weight ($\sum_i w_i = n$) \\
$F$ & Population distribution \\
$\hat{F}_w$ & Survey-weighted empirical distribution \\
$T(F)$ & Target functional (estimand) \\
$\hat{\theta} = T(\hat{F}_w)$ & Plug-in estimate \\
$\IF_i = \IF(x_i; T, F)$ & Influence function value for observation $i$ \\
$\hat{W} = \sum_i w_i = n$ & Normalized-weight total \\
$\hat{N} = \sum_i w_i^*$ & Estimated finite-population size (raw weights) \\
$\psi_i = w_i \IF_i / \hat{W}$ & Weighted linearized variable \\
\bottomrule
\end{tabular}
\end{table}

% ======================================================================
\section{Design-Based Variance for Modern DiD Estimators}
\label{sec:theory}

This section presents the paper's core contribution: a formal argument that
design-based variance estimation is valid for modern heterogeneity-robust DiD
estimators. The argument proceeds by showing that these estimators satisfy the
conditions of \citet{binder1983}, so that standard stratified-cluster variance
formulas applied to their influence function values produce design-consistent
variance estimators.

\subsection{Influence functions and smooth functionals}
\label{sec:if-smooth}

The influence function (IF) of a functional $T$ at distribution $F$ is the
Gateaux derivative:
\begin{equation}
  \IF(x; T, F) = \lim_{\varepsilon \to 0}
    \frac{T\bigl((1-\varepsilon)F + \varepsilon \, \delta_x\bigr) - T(F)}
         {\varepsilon}.
  \label{eq:if-def}
\end{equation}
This is a property of the map $T$ and the distribution $F$. Crucially, the IF
does not depend on how the sample was drawn --- the same functional $T$ has the
same IF regardless of whether the data come from simple random sampling,
stratified sampling, or cluster sampling. The IF characterizes each
observation's first-order contribution to the estimator.

Each modern heterogeneity-robust DiD estimator can be written as
$\theta = T(F)$ for a smooth functional $T$ that admits an IF representation.
Table~\ref{tab:estimators} classifies the key estimators and the source of
their smoothness.

\begin{table}[ht]
\centering
\caption{Modern DiD estimators and their smoothness properties. Smoothness
  arguments and IF characterizations for each estimator appear in
  Appendix~\ref{app:if-details}.}
\label{tab:estimators}
\begin{tabular}{@{}lll@{}}
\toprule
Estimator & Functional form & Smoothness source \\
\midrule
CS (reg) & Population means in group-time cells & Smooth in moments \\
CS (dr/ipw) & + propensity score + outcome regression &
  \citet{santannazhao2020} \\
Sun--Abraham & Interaction-weighted regression & Implicit function thm \\
Imputation DiD & OLS on untreated + imputation &
  \citet{borusyak2024} Thm~3 \\
TWFE & Standard regression & Smooth in moments \\
\bottomrule
\end{tabular}
\end{table}

\subsection{Main result}
\label{sec:main-result}

We now state the main result formally. The argument combines two established
bodies of theory --- influence functions for smooth functionals and design-based
variance for complex surveys --- in a setting where they have not previously
been connected.

\begin{assumption}[Stratified multi-stage sampling]
\label{ass:survey}
The sample is drawn via stratified multi-stage cluster sampling with $H$
strata and $n_h \geq 2$ sampled PSUs per stratum. Inclusion probabilities
$\pi_i > 0$ are known for all sampled units.
\end{assumption}

\begin{assumption}[Smooth functional]
\label{ass:smooth}
The estimand $\theta = T(F)$ admits a continuous influence function
$\IF(x; T, F)$, and the plug-in estimator satisfies the linearization:
\begin{equation}
  T(\hat{F}_w) - T(F)
  = \sum_{i \in \mathcal{S}} \psi_i
  + o_p(n^{-1/2}),
  \label{eq:vonmises}
\end{equation}
where $\mathcal{S}$ denotes the set of sampled units and
$\psi_i = w_i \IF(x_i; T, F) / \hat{W}$ is the weighted linearized variable
with $\hat{W} = \sum_{i \in \mathcal{S}} w_i$.
\end{assumption}

\begin{proposition}[Design-consistent variance for DiD]
\label{prop:main}
Under Assumptions~\ref{ass:survey} and~\ref{ass:smooth}, the design-based
variance of $\hat{\theta} = T(\hat{F}_w)$ is consistently estimated by
\begin{equation}
  \hat{V} = \sum_{h=1}^{H} (1 - f_h) \frac{n_h}{n_h - 1}
    \sum_{j=1}^{n_h} (\psi_{hj} - \bar{\psi}_{h\cdot})^2,
  \label{eq:binder-var}
\end{equation}
where $\psi_{hj} = \sum_{i \in \text{PSU}_{hj}} \psi_i$ is the PSU-level
total of influence function values and
$\bar{\psi}_{h\cdot} = n_h^{-1} \sum_{j=1}^{n_h} \psi_{hj}$ is the
within-stratum mean.
\end{proposition}

\begin{proof}[Proof sketch]
The argument proceeds in four steps.

\emph{Step 1: Modern DiD estimators are smooth functionals.}
As shown in Table~\ref{tab:estimators} and detailed in
Appendix~\ref{app:if-details}, each estimator reduces to combinations of
weighted means, regression coefficients, and smooth transformations thereof.
Each admits an IF representation, satisfying Assumption~\ref{ass:smooth}.

\emph{Step 2: IFs do not depend on the sampling design.}
The IF is a property of the map $T$ and the population distribution $F$
(equation~\ref{eq:if-def}). It does not depend on how the sample was drawn.

\emph{Step 3: Under survey weighting, the same IF form applies.}
Under survey weighting, the von Mises expansion
(equation~\ref{eq:vonmises}) expresses $\hat{\theta} - \theta$ as a weighted
sum of IF values over sampled observations. The variance of $\hat{\theta}$ is
therefore determined by the sampling design applied to these IF values.

\emph{Step 4: Binder's result.}
\citet{binder1983} showed that for any smooth functional satisfying
Assumption~\ref{ass:smooth}, applying the standard stratified-cluster
variance formula to the per-unit IF values produces a design-consistent
variance estimator. Equation~\eqref{eq:binder-var} is exactly this formula.
\end{proof}

\begin{remark}[Scope of the contribution]
\label{rem:scope}
Proposition~\ref{prop:main} is a direct application of \citet{binder1983} to
an arbitrary smooth functional. The substantive contribution of this paper is
not the proposition itself but the verification that each modern DiD estimator
in Table~\ref{tab:estimators} satisfies Assumption~\ref{ass:smooth}, and the
demonstration through simulation (Section~\ref{sec:simulation}) that the
resulting variance estimator is well-calibrated in finite samples. This
verification rests on the influence function representations established in the
original papers cited in Appendix~\ref{app:if-details}.
\end{remark}

\begin{remark}[Regularity conditions]
\label{rem:regularity}
The regularity conditions for Assumption~\ref{ass:smooth} are satisfied by each
estimator class through established results. For regression-based estimators
(TWFE, Sun--Abraham), the conditions follow from the implicit function theorem
applied to the estimating equations. For doubly robust estimators
(Callaway--Sant'Anna with DR), they follow from the semiparametric theory of
\citet{santannazhao2020}, though we note that their derivation assumes iid
sampling; the extension to survey-weighted nuisance estimation is supported by
our simulation evidence (Section~\ref{sec:sim-s4}) but not formally derived
here. For imputation estimators (Borusyak--Jaravel--Spiess), the IF from
Theorem~3 of \citet{borusyak2024} satisfies these conditions.
\end{remark}

\subsection{Variance computation}
\label{sec:variance-computation}

Proposition~\ref{prop:main} provides the theoretical foundation. In practice,
the variance is computed through one of two operational paths, depending on the
estimator's structure.

\paragraph{Regression-based TSL sandwich.}
For regression-based estimators (TWFE, Sun--Abraham, Stacked DiD), the
variance-covariance matrix is the stratified-cluster sandwich
\citep{binder1983}:
\begin{equation}
  V_{\text{TSL}} = (X'WX)^{-1}
    \left[\sum_{h=1}^{H} V_h \right]
    (X'WX)^{-1},
  \label{eq:sandwich}
\end{equation}
where $W = \diag(w_1, \ldots, w_n)$ and the stratum-level meat is
\begin{equation}
  V_h = (1 - f_h) \frac{n_h}{n_h - 1}
    \sum_{j=1}^{n_h} (T_{hj} - \bar{T}_{h\cdot})(T_{hj} - \bar{T}_{h\cdot})',
  \label{eq:meat}
\end{equation}
with $T_{hj} = \sum_{i \in \text{PSU}_{hj}} w_i X_i u_i$ the PSU-level
weighted score total ($u_i = Y_i - X_i'\hat{\beta}$ is the residual).

\paragraph{IF-based TSL.}
For IF-based estimators (Callaway--Sant'Anna, imputation DiD, two-stage DiD),
the variance is computed directly from the weighted linearized variables
$\psi_i$ without the bread matrix, using equation~\eqref{eq:binder-var}. The
estimator first computes per-unit influence function values $\IF_i$ for each
group-time cell, forms the weighted linearized variables
$\psi_i = w_i \IF_i / \hat{W}$, aggregates across cells, and then applies the
stratified-cluster formula.

\paragraph{Degrees of freedom.}
Inference uses the $t$-distribution with survey degrees of freedom
$\text{df} = \sum_h n_h - H$, i.e., the total number of sampled PSUs minus
the number of strata. For designs with replicate weights, the degrees of
freedom are computed from the rank of the replicate weight matrix, matching the
convention of R's \texttt{survey} package \citep{lumley2004}.

\paragraph{Replicate weight variance.}
An alternative to TSL is replicate weight variance, which perturbs the weights
and observes the resulting variation in estimates rather than linearizing the
estimator. The general form is
$V_{\text{rep}} = c \sum_r s_r (\hat{\theta}_r - \hat{\theta}_{\text{center}})^2$,
where $\hat{\theta}_r$ is the estimate under replicate $r$ and $c$ and $s_r$
are method-specific constants. This approach is useful when pre-computed
replicate weights are provided with public-use survey files (e.g., ACS PUMS,
CPS, BRFSS, RECS); for surveys without packaged replicate weights ---
including NHANES, where the CDC documentation recommends Taylor series
linearization with the masked strata and PSU variables --- TSL is the
canonical path. Five standard methods --- balanced repeated replication
(BRR), Fay's BRR, jackknife (JK1, JKn), and successive difference replication
(SDR) --- are detailed in Appendix~\ref{app:replicate}.

\subsection{Limitations of survey weighting for DiD}
\label{sec:limitations}

Before proceeding to empirical validation, we clarify the boundaries of the
design-based framework. Survey weighting addresses the sampling design. It does
not resolve every threat to valid causal inference with DiD. Practitioners
should be aware of the
following limitations.

\paragraph{Parallel trends.}
Survey weighting ensures that treatment effect estimates target the correct
population. It does not validate the parallel trends assumption. Under the
design-based framework, parallel trends must hold for the population, not just
the sample. If parallel trends fail in the population, survey-weighted
estimates remain biased --- with correctly estimated standard errors around the
wrong estimand.

\paragraph{Small-cluster asymptotics.}
TSL variance requires at least 2 PSUs per stratum ($n_h \geq 2$). With few
PSUs per stratum --- common in some state-based surveys --- the
$t$-distribution approximation with $\text{df} = \sum_h n_h - H$ may be
anti-conservative. Practitioners should assess the survey degrees of freedom
directly.

\paragraph{Estimand and weights.}
Survey weighting is necessary for the estimator to target the finite-population
ATT rather than the sample ATT --- this is precisely the correction demonstrated
in Scenario~2 of the simulation study. However, when different weighting schemes
are available (e.g., person-level vs.\ household-level weights, or calibrated
vs.\ design weights), the choice of weights determines which population
parameter is targeted \citep{solon2015}. The design-based variance estimator is
consistent for the variance of whatever the weighted estimator targets, but
practitioners should ensure the weighting scheme corresponds to the intended
estimand.

\paragraph{Weight variability.}
Highly variable weights reduce effective sample size. The Kish design effect
due to unequal weighting, $\DEFF_w = n \sum w_i^2 / (\sum w_i)^2$, measures
this: when $\DEFF_w \gg 1$, estimates are less precise than the nominal sample
size suggests.

\paragraph{Model misspecification.}
For doubly robust and IPW estimators, the IF corrections for propensity score
and outcome regression uncertainty assume correct specification of at least one
nuisance model. Survey weighting does not rescue a badly specified propensity
score or outcome model.

% ======================================================================
\section{Simulation Study}
\label{sec:simulation}

We assess the practical consequences of ignoring survey design in DiD
estimation through a Monte Carlo study with four scenarios, each targeting a
distinct aspect of the design-based variance framework developed in
Section~\ref{sec:theory}.

\subsection{Design of the simulation study}
\label{sec:sim-design}

\paragraph{Data-generating process.}
We generate staggered DiD panel data with stratified multi-stage survey
structure using a purpose-built DGP. The outcome for unit $i$ in period $t$ is:
\begin{equation}
  Y_{it} = \alpha_i + \gamma_{p(i)} + \eta_{p(i),t} + 0.5t
  + \beta_1 x_{1i} + \beta_2 x_{2i} + \tau_{it} + \varepsilon_{it},
  \label{eq:dgp}
\end{equation}
where $\alpha_i \sim N(0, \sigma^2_\alpha)$ is a unit fixed effect,
$\gamma_{p(i)} \sim N(0, \sigma^2_\gamma)$ is a PSU random effect for the PSU
$p(i)$ to which unit $i$ belongs, $\eta_{p(i),t} \sim N(0, \sigma^2_\eta)$ is
a PSU-by-period shock, $x_{1i} \sim N(0,1)$ and $x_{2i} \sim
\text{Bernoulli}(0.5)$ are covariates with coefficients $\beta_1 = 0.5$ and
$\beta_2 = 0.3$, $\tau_{it}$ is the treatment effect (zero for untreated
observations), and $\varepsilon_{it} \sim N(0, \sigma^2_\varepsilon)$ is
idiosyncratic noise. The PSU random effects $\gamma_{p(i)}$ create intra-cluster correlation in
outcome levels (relevant for cross-sectional estimators and the
within-transformation); the PSU-period shocks $\eta_{p(i),t}$ are the primary
driver of design effects in panel DiD, as they survive time-differencing and
inflate design-based SEs relative to HC1 SEs.

The survey structure consists of $H = 5$ strata with $n_h = 8$ sampled PSUs
per stratum and finite population correction $N_h = 200$. Sampling weights are
stratum-based with dispersion controlled by the coefficient of variation (CV).
Treatment is staggered across two cohorts (first treated at periods 3 and 5)
with a never-treated comparison group.

\paragraph{Estimators.}
We compare four estimators spanning the two variance paths from
Section~\ref{sec:variance-computation}:

\begin{itemize}
\item \textbf{Callaway--Sant'Anna, regression} (CS-reg): IF-based, no
  covariates. The workhorse modern DiD estimator
  \citep{callawaysantanna2021}.
\item \textbf{Callaway--Sant'Anna, doubly robust} (CS-DR): IF-based, with
  covariates $x_1$ and $x_2$ \citep{santannazhao2020}.
\item \textbf{Sun--Abraham} (SA): Regression-based, interaction-weighted
  \citep{sunabraham2021}.
\item \textbf{Two-way fixed effects} (TWFE): The traditional regression
  approach, included as a baseline to illustrate the bias that
  heterogeneity-robust estimators correct.
\end{itemize}

For each estimator and replication, we compute the ATT and standard error
under three inference approaches: (i)~\emph{HC1}, using
heteroskedasticity-robust standard errors with no clustering, no weighting,
and no survey design --- the default in most DiD software; (ii)~\emph{PSU
cluster}, applying the survey-weighted point estimate with PSU-level
clustering only (no strata, no FPC) --- the practitioner heuristic of
\texttt{cluster=psu} combined with survey weights; and
(iii)~\emph{design-based}, using TSL with strata, PSU, FPC, and survey
weights as developed in Section~\ref{sec:theory}. Approaches (ii) and (iii)
share the same survey-weighted point estimate; they differ only in the
variance estimator.

\paragraph{Metrics.}
For each cell (scenario $\times$ estimator $\times$ sample size), we report:
\emph{Bias} (mean estimate minus true ATT);
\emph{Coverage} (fraction of 95\% CIs containing the true ATT);
\emph{DEFF} (ratio of design-based to HC1 variance,
$\hat{V}_\text{design} / \hat{V}_\text{HC1}$); and
\emph{RMSE} (root mean squared error).
Each cell consists of 2,000 independent replications at sample sizes
$n \in \{500, 2{,}000, 8{,}000\}$, yielding 66,000 total replications across
the study. The four estimators span both variance paths from
Section~\ref{sec:variance-computation}; simulation validation of the remaining
estimators in Table~\ref{tab:estimator-coverage} is left for future work.
With 2,000 replications, the Monte Carlo standard error for a coverage
estimate near the nominal 95\% level is approximately $\sqrt{0.95 \times
0.05 / 2000} \approx 0.5$ percentage points, so reported coverage values
within roughly $\pm$1 percentage point of nominal should be interpreted as
indistinguishable from $0.95$ at conventional simulation precision.

\subsection{Scenario 1: Unconditional parallel trends with complex survey}
\label{sec:sim-s1}

The baseline scenario generates data under standard parallel trends with
moderate intra-class correlation ($\text{ICC} = 0.1$) and weight variation
($\text{CV} = 0.5$), producing a realistic survey design where design effects
are present but not extreme. Table~\ref{tab:sim-s1} reports the results.

\input{tables/sim_s1}

\begin{figure}[ht]
\centering
\includegraphics[width=\textwidth]{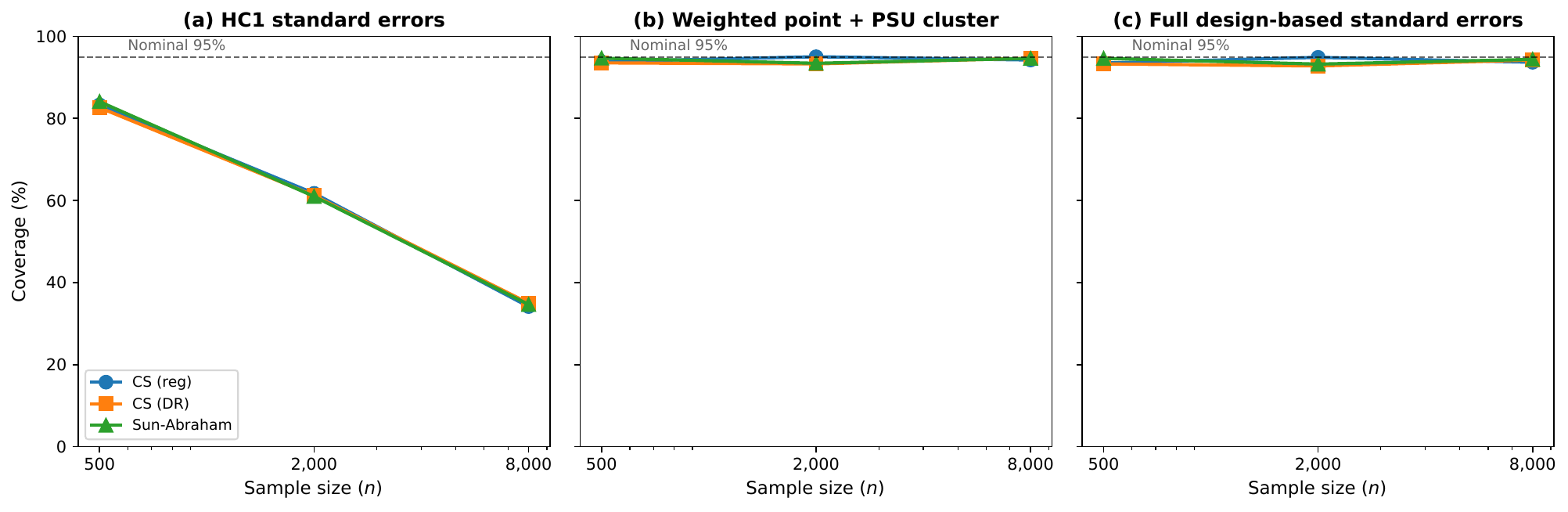}
\caption{95\% confidence interval coverage for modern DiD estimators under
  complex survey design (Scenario~1). Panel~(a): HC1 standard errors that
  ignore the survey design entirely. Panel~(b): the survey-weighted point
  estimate combined with PSU-level clustering only (no strata, no FPC) ---
  the practitioner heuristic. Panel~(c): full design-based standard errors
  from Proposition~\ref{prop:main}. The dashed line marks the nominal
  95\% level. HC1 coverage collapses as sample size grows; both PSU-clustered
  and full design-based coverage remain near nominal across all estimators
  and sample sizes. TWFE (not shown) has 0\% coverage at all sample sizes
  due to estimand bias under staggered treatment. Based on 2,000
  replications per cell.}
\label{fig:coverage-s1}
\end{figure}

The core finding is stark: \textbf{HC1 standard errors produce dramatically
under-covered confidence intervals under complex survey designs, and the
problem worsens with larger samples.} At $n = 500$, HC1 coverage for CS-reg
is 83.3\% --- already below the nominal 95\% --- but at $n = 8{,}000$ it
collapses to 34.2\%. This is not a small-sample artifact; it reflects the
fundamental mismatch between the iid assumption embedded in HC1 SEs and the
stratified-cluster structure of the data. The estimated design effect grows from
2.1$\times$ at $n = 500$ to 17.3$\times$ at $n = 8{,}000$, as larger samples
make the cluster structure increasingly salient relative to the within-cluster
variation.

Design-based SEs from Proposition~\ref{prop:main} correct this: coverage is
92.8--95.0\% across all modern estimators and sample sizes, closely
approximating the nominal 95\% level. The correction is consistent across
CS-reg, CS-DR, and Sun--Abraham, confirming that the theoretical framework
applies uniformly to both IF-based and regression-based variance paths.

\paragraph{Does PSU clustering alone suffice?}
A reader familiar with cluster-robust inference might ask whether simply
clustering at the PSU level --- without strata or finite population
corrections --- is enough once the point estimate is properly survey-weighted.
The ``Cluster'' coverage column in Table~\ref{tab:sim-s1} reports this
directly: the survey-weighted point estimate (which kills any informative-sampling
bias) combined with PSU-level clustering only (no strata, no FPC). In
Scenario~1, PSU clustering matches full design-based coverage at every
sample size --- 93.8\%, 95.0\%, and 94.2\% for CS-reg, versus 93.5\%,
94.9\%, and 93.7\% under the full design. Strata and FPC contribute little
to coverage once PSU clustering and survey weights are in place. As
Scenarios~2--4 will show, the same pattern holds under informative sampling
and conditional parallel trends.

TWFE shows 0\% coverage at every sample size --- not because of the survey
design, but because it is biased under staggered treatment adoption
\citep{dechaisemartin2020}. The ATT estimate of 0.31 versus the true value of
2.0 reflects the well-documented negative weighting problem. Design-based SEs
correctly estimate the variance of this biased estimand but cannot rescue the
point estimate.

\subsection{Scenario 2: Informative sampling with heterogeneous treatment effects}
\label{sec:sim-s2}

This scenario adds two complications: (i)~sampling weights are correlated with
potential outcomes (informative sampling), so high-outcome units receive heavier
weights; and (ii)~treatment effects vary across strata (ranging from 0.6 to
3.4, with a population-weighted ATT of approximately 2.6). Under informative
sampling, unweighted estimators target the sample ATT rather than the
population ATT. Table~\ref{tab:sim-s2} reports the results.

\input{tables/sim_s2}

The unweighted (no-survey-design) estimates are biased by approximately
$-0.19$ across all modern estimators --- they systematically underestimate
the population ATT because the sample over-represents strata with smaller
treatment effects. Design-based (survey-weighted) estimates eliminate this
bias, with mean bias near zero for CS-reg, CS-DR, and Sun--Abraham.

HC1 coverage is severely depressed (10--44\%), while design-based coverage
remains near 95\%. The design effects are large (8--136$\times$), driven by
the high weight variation in this scenario. The corrected practitioner
heuristic --- weighted point estimate plus PSU-level clustering --- also
maintains near-nominal coverage (93--97\% across estimators), and is
slightly conservative for CS-reg and CS-DR (96--97\% versus the full
design's 94\%). The over-coverage of the cluster specification reflects the
absence of stratification and finite-population reduction, both of which
shrink the variance toward its design-consistent target. The substantive
takeaway is the same as in Scenario~1: survey weighting is doing the
estimand-correction work and PSU clustering is doing the variance work;
strata and FPC tighten the variance further but are not required for
nominal coverage.

\subsection{Scenario 3: Repeated cross-sections}
\label{sec:sim-s3}

Many federal surveys (BRFSS, NHANES, MEPS) are repeated cross-sections rather
than panels. This scenario verifies that the design-based framework extends to
this setting, using CS-reg with $\texttt{panel} = \texttt{False}$.
Table~\ref{tab:sim-s3} reports the results.

\input{tables/sim_s3}

Design-based coverage is 93.5--94.7\% across all sample sizes, confirming that
Proposition~\ref{prop:main} applies to repeated cross-sectional survey data.
HC1 coverage degrades from 94.5\% at $n = 500$ to 69.5\% at $n = 8{,}000$,
following the same pattern as the panel case in Scenario~1 (though with smaller
design effects because the cross-sectional DGP has less persistent clustering).

\subsection{Scenario 4: Conditional parallel trends}
\label{sec:sim-s4}

This scenario examines whether survey-weighted doubly robust estimation ---
where both the propensity score and outcome regression are estimated with
survey weights --- produces valid inference under conditional parallel trends.

The DGP adds an X-dependent time trend: the outcome includes a term
$\delta \cdot x_{1i} \cdot (t / T)$ where $\delta = 1.5$, and treated units'
$x_1$ is drawn from $N(1, 1)$ rather than $N(0, 1)$. Because
$\E[x_1 \mid \text{treated}] \neq \E[x_1 \mid \text{control}]$, the average
time trend differs by group, violating unconditional parallel trends.
Conditional on $x_1$, trends are identical --- conditional parallel trends
holds. Table~\ref{tab:sim-s4} reports the results.

\input{tables/sim_s4}

CS-reg, which does not condition on covariates, exhibits large and persistent
bias ($+0.57$, roughly 29\% of the true ATT) with 0\% coverage at every sample
size. The bias does not diminish with $n$ because it reflects a
misspecification of the identifying assumption, not sampling variability.

CS-DR, which conditions on $x_1$ and $x_2$ via doubly robust estimation with
survey-weighted nuisance models, eliminates this bias: mean bias is near zero
at $n = 2{,}000$ with 94.7\% coverage. The RMSE ratio between the two
approaches is 8.2$\times$.

These simulation results provide strong evidence that the design-based variance
framework extends to doubly robust estimation with survey-weighted nuisance
models. A complete formal derivation would require showing that the
\citet{santannazhao2020} influence function expansion carries through when the
propensity score and outcome regression are estimated with survey-weighted
estimating equations --- a nontrivial extension of their iid theory that we
leave for future work. The simulation evidence here, with coverage at
92.7--94.7\% across 2,000 replications and three sample sizes, suggests that
the composition of survey-weighted nuisance estimation with design-based
variance is well calibrated in this DGP. The practical implication is that practitioners working with survey
data where parallel trends hold only after covariate adjustment can use
design-based DR estimation with well-calibrated inference.

% ======================================================================
\section{Empirical Illustration: NHANES and the ACA}
\label{sec:empirical}

We illustrate the practical importance of design-based inference with an
analysis of the Affordable Care Act's (ACA) dependent coverage provision using
data from the National Health and Nutrition Examination Survey (NHANES).

\subsection{Policy setting}

The ACA's dependent coverage mandate, effective September 2010, allowed young
adults to remain on their parents' health insurance until age 26. This created a
natural experiment: adults aged 19--25 gained access to dependent coverage
(treatment group), while adults aged 27--34, who were demographically similar
but ineligible, serve as a comparison group. This age-cutoff DiD design was
established by \citet{antwi2013} and has been widely used in evaluations of ACA
coverage effects \citep{sommers2012}.

Our outcome is binary health insurance coverage, estimated via a linear
probability model. We use two NHANES cycles: 2007--2008 (pre-ACA) and
2015--2016 (post-ACA), yielding $n = 2{,}946$ observations (1,372 in the
treatment group, 1,574 in the comparison group).

\subsection{NHANES survey design}

NHANES employs a stratified multi-stage cluster design to produce nationally
representative health and nutrition data \citep{nchs2018}. The public-use files
provide masked survey design variables: 31 pseudo-strata (SDMVSTRA), 2 PSUs per
stratum (SDMVPSU), and 2-year examination weights (WTMEC2YR). PSU identifiers
are reused across strata, requiring the nested specification when constructing
the survey design object. We follow CDC's analytic guidelines for proper
variance estimation with these design variables.\footnote{The analytic sample
is restricted to MEC-examined respondents (those with $\text{WTMEC2YR} > 0$),
and we use \texttt{WTMEC2YR} as the appropriate weight for that domain. CDC's
\emph{Weighting Module} convention selects the sample weight by the data
source of the smallest-respondent-set variable used; the MEC weight applies
when the analysis is restricted to (or includes any variables from) the MEC
examination component. An alternative analysis using \texttt{WTINT2YR} over
the full interview sample would target a broader population --- adults who
completed only the in-home interview as well as MEC-examined adults. We
expect the qualitative lesson --- that weighting and clustering materially
affect inference --- to remain, but this robustness check is left for future
replication.}

\subsection{Results}

Table~\ref{tab:nhanes} presents DiD estimates under four specifications that
progressively incorporate survey design features.

\input{tables/nhanes}

Three findings emerge. First, \textbf{survey weighting changes the point
estimate substantially}: the unweighted ATT is 6.5 percentage points, while the
survey-weighted ATT is 9.6 percentage points --- a 48\% increase. This
difference arises because NHANES's sampling design over-represents certain
demographic groups. Without weighting, the estimated treatment effect reflects
the sample composition rather than the U.S.\ population of young adults.

Second, \textbf{ignoring PSU clustering understates uncertainty}. Comparing the
weights-only specification (SE $= 0.037$) with the full design specification
(SE $= 0.044$), clustering inflates the standard error by 19\% (a design effect
of approximately 1.4$\times$ on the variance scale). This is modest compared
to the simulation study in Section~\ref{sec:simulation} (where design effects
reached 17$\times$ in the baseline scenario and exceeded 100$\times$ under
informative sampling), reflecting the particular variance structure of this
analytic sample, which has only 2 PSUs per stratum. Even so, the 95\% confidence interval shifts from
$[0.023, 0.170]$ to $[0.006, 0.187]$ --- the lower bound moves substantially
closer to zero.

Third, \textbf{the HC1 specification produces a qualitatively different
conclusion}: the unweighted, no-design ATT of 0.065 (SE $= 0.035$) has a
95\% CI of $[-0.002, 0.133]$ that includes zero, while the design-based ATT
of 0.096 (SE $= 0.044$) has a CI of $[0.006, 0.187]$ that excludes zero.
This significance reversal is driven primarily by the change in the point
estimate from proper survey weighting (0.065 to 0.096), not by the variance
correction --- the design-based SE is actually \emph{larger} than the HC1
SE, which widens the confidence interval. Both corrections --- weighting
the point estimate upward and widening the confidence interval via
design-based SEs --- are needed for valid population-level inference,
though they have opposing effects on the test statistic. In this case, the
point estimate correction dominates and flips the significance conclusion.

Covariate adjustment for gender and family income-to-poverty ratio reduces the
ATT to 6.0 percentage points with a CI that includes zero, suggesting that some
of the measured coverage gain is confounded with demographic shifts between the
pre- and post-ACA periods. This underscores the importance of conditional
parallel trends analysis in survey-weighted DiD --- exactly the setting
addressed by Scenario~4 of the simulation study.

A caveat on the empirical design: our post-period (2015--2016) follows not only
the 2010 dependent coverage mandate but also the major 2014 ACA provisions
(Medicaid expansion and health insurance exchanges), which likely had
differential effects on the treatment and control age groups. Standard
evaluations of the dependent coverage mandate \citep{antwi2013} use post-periods
prior to 2014 to avoid this confound. Our purpose here is to illustrate the
methodological consequences of ignoring the survey design, not to make a causal
claim about the dependent coverage mandate itself.

% ======================================================================
\section{Software}
\label{sec:software}

The theoretical framework and empirical results in this paper are implemented in
\texttt{diff-diff} (v3.3.2; \citealp{diffdiff2026}), an open-source Python
package for difference-in-differences causal inference. The package provides
design-based variance estimation for 15 modern DiD estimators. To our
knowledge, this is the first open-source implementation that jointly
supports strata, PSU clustering, finite population corrections, and
replicate-weight methods for the heterogeneity-robust estimators discussed
in Section~\ref{sec:theory}.

\subsection{API design}

The package follows an sklearn-like API where estimators are instantiated with
configuration parameters and fitted to data via a \texttt{fit()} method. Survey
design information is encapsulated in a \texttt{SurveyDesign} object that
specifies the design variables and is passed to \texttt{fit()} as an optional
argument. This design keeps the survey specification separate from the
estimator logic, so the same estimator code works with or without survey
adjustments:

\begin{verbatim}
from diff_diff import CallawaySantAnna, SurveyDesign

design = SurveyDesign(
    weights="w", strata="strata",
    psu="psu", nest=True
)

cs = CallawaySantAnna(estimation_method="dr")
result = cs.fit(
    data, outcome="y", unit="id",
    time="period", first_treat="first_treat",
    covariates=["x1", "x2"],
    survey_design=design
)
print(result.summary())
\end{verbatim}

\subsection{Estimator coverage}

Table~\ref{tab:estimator-coverage} summarizes the survey variance support for
each estimator class. Regression-based estimators use the TSL sandwich
(equation~\ref{eq:sandwich}); IF-based estimators apply the stratified-cluster
formula directly to influence function values (equation~\ref{eq:binder-var}).
Two estimators (Synthetic DiD, TROP) involve non-smooth optimization and
rely on resampling for survey variance: Synthetic DiD supports Rao--Wu
rescaled bootstrap, stratified placebo permutation, and PSU-level jackknife;
TROP supports Rao--Wu rescaled bootstrap only.

\begin{table}[ht]
\centering
\small
\caption{Survey variance support by estimator.}
\label{tab:estimator-coverage}
\begin{tabular}{@{}llll@{}}
\toprule
Estimator & Variance path & TSL & Replicate \\
\midrule
Callaway--Sant'Anna \citep{callawaysantanna2021} & IF-based & Yes & Yes \\
Sun--Abraham \citep{sunabraham2021} & Regression sandwich & Yes & Yes \\
Imputation DiD \citep{borusyak2024} & IF-based (refit) & Yes & Yes \\
Two-Stage DiD \citep{gardner2022} & IF-based (refit) & Yes & Yes \\
TWFE & Regression sandwich & Yes & Yes \\
Stacked DiD \citep{cengiz2019} & Regression sandwich & Yes & Yes \\
Efficient DiD \citep{chensantannaxie2025} & IF-based (EIF) & Yes & Yes \\
Continuous DiD \citep{callaway2024continuous} & Regression sandwich & Yes & Yes \\
Triple Difference & IF-based & Yes & Yes \\
Staggered Triple Diff & IF-based & Yes & Yes \\
dCDH \citep{dechaisemartin2020} & IF-based & Yes & No \\
Heterogeneous Adoption DiD \citep{dechaisemartin2026had} & IF-based & Yes & No \\
Wooldridge DiD \citep{wooldridge2021} & Regression sandwich & Yes & No \\
Synthetic DiD & Resampling & No & No \\
TROP \citep{athey2025trop} & Bootstrap only & No & No \\
\bottomrule
\end{tabular}
\\[2pt]\footnotesize
\emph{Notes:} ``TSL'' and ``Replicate'' columns refer to analytical and
replicate-weight variance under the Binder-style framework of
Section~\ref{sec:theory}. Synthetic DiD and TROP rely on resampling for
survey variance and so report ``No'' in both analytical columns; their
survey-design support is provided through Rao--Wu rescaled bootstrap (and,
for Synthetic DiD, stratified placebo permutation and PSU-level jackknife),
which is not captured by these two columns.
\end{table}

\subsection{Availability}

\texttt{diff-diff} is available on PyPI (\texttt{pip install diff-diff}) and
GitHub. Documentation, tutorials, and the survey tutorial that accompanies this
paper are available at the project website.

% ======================================================================
\section{Discussion}
\label{sec:discussion}

This paper bridges the survey statistics and modern DiD literatures by showing
that Binder's (1983) design-based variance theory applies directly to the
influence functions of heterogeneity-robust DiD estimators. The result is
simple in retrospect --- the regularity conditions are met, the formulas are
known --- but it has not previously been stated for this class of estimators.
Empirically, the framework recovers the practitioner heuristic that already
sees wide informal use --- weighted point estimates with PSU-level clustering
--- and identifies the additional precision available from the full
strata/FPC specification.

\subsection{Practical guidance}
\label{sec:practical-guidance}

The simulation results in Section~\ref{sec:simulation} resolve cleanly into a
short set of recommendations. The first two steps are necessary; the third
is a precision improvement on top of an already-valid specification.

\paragraph{Always weight the point estimate.}
When the target is a population-level treatment effect, use survey weights for
estimation regardless of which variance specification you choose. The NHANES
illustration in Section~\ref{sec:empirical} showed a 48\% difference in the
point estimate between weighted and unweighted specifications --- a gap that
no variance estimator can recover. Weights affect what the estimand
\emph{is}, not just how precisely it is measured. Under informative sampling,
the bias from skipping weights is large (see Scenario~2, where unweighted
estimators have RMSE roughly $1.1\times$ the weighted RMSE despite identical
sample sizes).

\paragraph{Cluster at the PSU level for the variance.}
For variance, PSU-level clustering through the influence-function
representation --- a \texttt{SurveyDesign} with \texttt{psu} and survey
weights but no strata or FPC --- is sufficient for nominal coverage in every
scenario we examined, including under informative sampling. This matches the
\texttt{cluster=psu} heuristic that many applied practitioners already use, and Section~\ref{sec:theory} provides the formal justification
that has been missing for the modern heterogeneity-robust estimators.

\paragraph{Add strata and FPC for precision, not coverage.}
The full design specification (strata, PSU clustering, FPC, and weights)
narrows the variance modestly relative to the cluster-only specification,
but the gain is incremental. In our Scenario~2 (informative sampling, design
effects up to $136\times$), the cluster-only specification was 1--3
percentage points more conservative than the full design --- the type of
over-coverage that costs precision in these simulations but does not
compromise validity. Practitioners who
have access to the full design variables should use them. Practitioners
working with public-use files where strata or FPC are masked or suppressed
can use the cluster-only specification, which appears conservative for
coverage in the designs we study; full design variables remain preferred
when available, particularly when stratification is fine-grained or
finite-population corrections are large.

\paragraph{Replicate weights are an equivalent path.}
Some federal surveys (BRFSS, ACS, RECS) ship replicate-weight files designed
to capture the full design variance without requiring strata and PSU
identifiers, which are sometimes suppressed for confidentiality.
Replicate-weight methods (BRR, Fay's BRR, JK1, JKn, SDR) applied to the
influence-function values produce variance estimates asymptotically equivalent
to the TSL formula. Practitioners with replicate weights should use them
directly rather than reconstructing strata.

\paragraph{Use the design effect as a sanity check.}
If the estimated DEFF (ratio of design-based to HC1 variance) is close to
unity, the survey design has little impact and reporting HC1 alongside
design-based SEs is reasonable. If DEFF exceeds two --- common in health
and social surveys with intra-cluster correlation --- HC1 inference is
materially misleading and design-based inference (cluster or full) should
be primary.

\subsection{Limitations}

Several limitations should be noted. The TSL variance formula requires at
least two PSUs per stratum ($n_h \geq 2$); designs with singleton strata
require special handling. The $t$-distribution approximation with
$\text{df} = \sum_h n_h - H$ degrees of freedom may be anti-conservative when
the total number of PSUs is small (yielding few overall degrees of freedom),
a common concern in state-based surveys. Our simulation study used a balanced design ($H = 5$ strata, $n_h = 8$
PSUs per stratum) with a fixed number of PSUs as the sample size $n$ increases.
This diverges from the standard PSU-based asymptotic sequence in which both $n$
and the number of PSUs grow together; our simulation instead increases the
within-PSU sample size at fixed design complexity. The performance of
design-based SEs under severely unbalanced designs, alternative asymptotic
regimes, or very few PSUs per stratum merits further investigation. When the
number of PSUs per stratum is small, finite-sample degrees-of-freedom
corrections such as the bias-reduced linearization of \citet{bellmccaffrey2002}
or the small-sample adjustments of \citet{faygraubard2001} can improve
coverage; incorporating these into the design-based DiD framework is a natural
next step.

A related practical issue is the level of clustering. The analytical framework
aggregates to survey PSUs, but policy evaluations often apply treatments at
higher geographic units (states, counties) that may cross-cut the survey
stratification. When the treatment assignment level differs from the survey
clustering level, reconciling the two sources of dependence requires additional
care beyond the scope of this paper.

The framework assumes that the parallel trends assumption holds in the
population. Survey weighting ensures that the estimand corresponds to the
correct population, but it does not validate the identifying assumption itself.
The conditional parallel trends result (Scenario~4) requires correct
specification of at least one nuisance model, as with any doubly robust
estimator.

Our empirical illustration uses a two-period, two-group design. While the
theoretical framework covers the full staggered adoption case, the real-data
example does not demonstrate this. Applications of design-based variance
estimation to staggered treatment designs using federal survey data (e.g.,
state-level policy variation in the CPS or ACS) would be a valuable complement.

\subsection{Future directions}

Several extensions merit investigation. First, calibration and
post-stratification weights --- increasingly common in modern survey practice
--- interact with the variance estimation framework in ways that are not fully
addressed here. Second, multi-level treatment designs, where treatment is
assigned at a level that cross-cuts the survey design (e.g., county-level
policies in a state-stratified survey), raise additional questions about the
appropriate variance estimator. Third, extending the analytical framework to
non-smooth estimators such as synthetic control methods would broaden the
applicability of design-based DiD inference. Finally, the growing availability
of replicate weight files with public-use survey data creates an opportunity
for practitioners to implement design-based variance without needing access to
the underlying strata and PSU identifiers, which are sometimes suppressed for
confidentiality.

\subsection*{Code and data availability}

The simulation pipeline (DGP, estimator wrappers, table and figure
generation) and the NHANES illustration are implemented in Python against
\texttt{diff-diff} v3.3.2. The replication code is publicly archived at
\url{https://github.com/igerber/design-based-did-replication}, and the
v1.0 release accompanying this preprint is deposited at Zenodo
\citep{gerberreplication2026} under DOI
\href{https://doi.org/10.5281/zenodo.20097361}{10.5281/zenodo.20097361};
the concept DOI
\href{https://doi.org/10.5281/zenodo.20097360}{10.5281/zenodo.20097360}
resolves to the latest version. The NHANES public-use files are available
from the U.S.\ National Center for Health Statistics at
\url{https://wwwn.cdc.gov/nchs/nhanes}; we ship the analysis-ready golden
JSON used to produce Table~\ref{tab:nhanes} so that the empirical
illustration can be reproduced without re-downloading the raw NHANES
cycles.

% ======================================================================
\bibliography{references}

% ======================================================================
\clearpage
\appendix
\input{appendix}

\end{document}

%% file: tables/sim_s1.tex
\begin{table}[ht]
\centering
\small
\caption{Simulation results: Unconditional PT + complex survey.}
\label{tab:sim-s1}
\begin{tabular}{@{}llrrrrrrrr@{}}
\toprule
 & & \multicolumn{2}{c}{Bias} & \multicolumn{3}{c}{Coverage (\%)} &  & \multicolumn{2}{c}{RMSE} \\
\cmidrule(lr){3-4} \cmidrule(lr){5-7} \cmidrule(lr){9-10}
{Estimator} & {$n$} & {Unwt} & {Wtd} & {HC1} & {Cluster} & {Design} & {DEFF} & {Unwt} & {Wtd} \\
\midrule
CS (reg) & 500 & -0.000 & -0.000 & 83.3 & 93.8 & 93.5 & 2.1 & 0.065 & 0.069 \\
 & 2,000 & +0.003 & +0.003 & 61.8 & 95.0 & 94.9 & 5.2 & 0.052 & 0.054 \\
 & 8,000 & -0.000 & -0.000 & 34.2 & 94.2 & 93.7 & 17.3 & 0.050 & 0.051 \\
CS (DR) & 500 & +0.000 & -0.001 & 82.7 & 93.5 & 93.3 & 2.1 & 0.067 & 0.071 \\
 & 2,000 & -0.000 & -0.000 & 61.2 & 93.3 & 92.8 & 5.1 & 0.054 & 0.056 \\
 & 8,000 & -0.001 & -0.001 & 35.0 & 94.7 & 94.2 & 17.4 & 0.050 & 0.050 \\
Sun-Abraham & 500 & +0.000 & +0.000 & 84.1 & 94.7 & 94.7 & 2.1 & 0.064 & 0.069 \\
 & 2,000 & +0.000 & +0.000 & 61.0 & 93.5 & 93.2 & 5.2 & 0.053 & 0.055 \\
 & 8,000 & -0.001 & -0.001 & 34.6 & 94.7 & 94.4 & 17.4 & 0.049 & 0.049 \\
TWFE & 500 & -1.686 & -1.686 & 0.0 & 0.0 & 0.0 & 1.8 & 1.686 & 1.686 \\
 & 2,000 & -1.686 & -1.686 & 0.0 & 0.0 & 0.0 & 4.4 & 1.686 & 1.686 \\
 & 8,000 & -1.686 & -1.686 & 0.0 & 0.0 & 0.0 & 14.8 & 1.686 & 1.686 \\
\bottomrule
\end{tabular}
\\[2pt]\footnotesize \emph{Notes:} HC1 uses the unweighted point estimate (Unwt) with heteroskedasticity-robust standard errors. Cluster and Design both use the survey-weighted point estimate (Wtd); Cluster applies PSU-level clustering only, while Design applies full Taylor-series linearization with strata, PSU, and FPC.
\end{table}

%% file: tables/sim_s2.tex
\begin{table}[ht]
\centering
\small
\caption{Simulation results: Informative sampling + heterogeneous TE.}
\label{tab:sim-s2}
\begin{tabular}{@{}llrrrrrrrr@{}}
\toprule
 & & \multicolumn{2}{c}{Bias} & \multicolumn{3}{c}{Coverage (\%)} &  & \multicolumn{2}{c}{RMSE} \\
\cmidrule(lr){3-4} \cmidrule(lr){5-7} \cmidrule(lr){9-10}
{Estimator} & {$n$} & {Unwt} & {Wtd} & {HC1} & {Cluster} & {Design} & {DEFF} & {Unwt} & {Wtd} \\
\midrule
CS (reg) & 500 & -0.193 & +0.002 & 43.5 & 96.5 & 94.7 & 7.9 & 0.325 & 0.291 \\
 & 2,000 & -0.194 & -0.009 & 22.7 & 96.2 & 94.0 & 30.7 & 0.329 & 0.294 \\
 & 8,000 & -0.183 & +0.008 & 11.2 & 96.4 & 94.3 & 122.2 & 0.318 & 0.287 \\
CS (DR) & 500 & -0.188 & +0.004 & 43.5 & 96.9 & 94.5 & 7.8 & 0.323 & 0.297 \\
 & 2,000 & -0.182 & +0.002 & 25.4 & 95.9 & 94.0 & 30.3 & 0.318 & 0.291 \\
 & 8,000 & -0.198 & -0.014 & 12.1 & 96.0 & 94.2 & 121.7 & 0.325 & 0.292 \\
Sun-Abraham & 500 & -0.190 & +0.002 & 43.1 & 94.7 & 94.2 & 8.3 & 0.321 & 0.290 \\
 & 2,000 & -0.196 & +0.000 & 20.4 & 93.5 & 93.0 & 32.3 & 0.327 & 0.295 \\
 & 8,000 & -0.191 & +0.004 & 10.3 & 95.3 & 94.8 & 127.2 & 0.320 & 0.284 \\
TWFE & 500 & -2.159 & -2.089 & 0.0 & 0.0 & 0.0 & 8.7 & 2.159 & 2.090 \\
 & 2,000 & -2.159 & -2.089 & 0.0 & 0.0 & 0.0 & 34.0 & 2.159 & 2.089 \\
 & 8,000 & -2.157 & -2.086 & 0.0 & 0.0 & 0.0 & 135.7 & 2.157 & 2.086 \\
\bottomrule
\end{tabular}
\\[2pt]\footnotesize \emph{Notes:} HC1 uses the unweighted point estimate (Unwt) with heteroskedasticity-robust standard errors. Cluster and Design both use the survey-weighted point estimate (Wtd); Cluster applies PSU-level clustering only, while Design applies full Taylor-series linearization with strata, PSU, and FPC.
\end{table}

%% file: tables/sim_s3.tex
\begin{table}[ht]
\centering
\small
\caption{Simulation results: Repeated cross-section.}
\label{tab:sim-s3}
\begin{tabular}{@{}llrrrrrrrr@{}}
\toprule
 & & \multicolumn{2}{c}{Bias} & \multicolumn{3}{c}{Coverage (\%)} &  & \multicolumn{2}{c}{RMSE} \\
\cmidrule(lr){3-4} \cmidrule(lr){5-7} \cmidrule(lr){9-10}
{Estimator} & {$n$} & {Unwt} & {Wtd} & {HC1} & {Cluster} & {Design} & {DEFF} & {Unwt} & {Wtd} \\
\midrule
CS (reg) & 500 & +0.001 & -0.001 & 94.5 & 94.9 & 94.7 & 1.3 & 0.116 & 0.129 \\
 & 2,000 & -0.002 & -0.003 & 87.5 & 94.7 & 94.5 & 1.8 & 0.073 & 0.077 \\
 & 8,000 & -0.002 & -0.002 & 69.5 & 94.1 & 93.5 & 3.9 & 0.055 & 0.057 \\
\bottomrule
\end{tabular}
\\[2pt]\footnotesize \emph{Notes:} HC1 uses the unweighted point estimate (Unwt) with heteroskedasticity-robust standard errors. Cluster and Design both use the survey-weighted point estimate (Wtd); Cluster applies PSU-level clustering only, while Design applies full Taylor-series linearization with strata, PSU, and FPC.
\end{table}

%% file: tables/sim_s4.tex
\begin{table}[ht]
\centering
\small
\caption{Simulation results: Conditional PT (headline).}
\label{tab:sim-s4}
\begin{tabular}{@{}llrrrrrrrr@{}}
\toprule
 & & \multicolumn{2}{c}{Bias} & \multicolumn{3}{c}{Coverage (\%)} &  & \multicolumn{2}{c}{RMSE} \\
\cmidrule(lr){3-4} \cmidrule(lr){5-7} \cmidrule(lr){9-10}
{Estimator} & {$n$} & {Unwt} & {Wtd} & {HC1} & {Cluster} & {Design} & {DEFF} & {Unwt} & {Wtd} \\
\midrule
CS (reg) & 500 & +0.580 & +0.569 & 0.0 & 0.0 & 0.0 & 1.5 & 0.588 & 0.577 \\
 & 2,000 & +0.582 & +0.570 & 0.0 & 0.0 & 0.0 & 3.7 & 0.586 & 0.574 \\
 & 8,000 & +0.581 & +0.570 & 0.0 & 0.0 & 0.0 & 12.3 & 0.585 & 0.574 \\
CS (DR) & 500 & -0.002 & -0.002 & 82.8 & 93.5 & 92.7 & 1.8 & 0.093 & 0.093 \\
 & 2,000 & -0.000 & +0.000 & 64.6 & 94.8 & 94.7 & 4.5 & 0.071 & 0.070 \\
 & 8,000 & +0.000 & -0.000 & 37.0 & 94.6 & 94.0 & 15.2 & 0.068 & 0.068 \\
\bottomrule
\end{tabular}
\\[2pt]\footnotesize \emph{Notes:} HC1 uses the unweighted point estimate (Unwt) with heteroskedasticity-robust standard errors. Cluster and Design both use the survey-weighted point estimate (Wtd); Cluster applies PSU-level clustering only, while Design applies full Taylor-series linearization with strata, PSU, and FPC.
\end{table}

%% file: tables/nhanes.tex
\begin{table}[ht]
\centering
\small
\caption{ACA dependent coverage provision: DiD estimates of the effect on health insurance coverage using NHANES data, 2007--2008 vs.\ 2015--2016. Treatment group: ages 19--25; control group: ages 27--34.}
\label{tab:nhanes}
\begin{tabular}{@{}lrrrr@{}}
\toprule
Specification & ATT & SE & 95\% CI & df \\
\midrule
Naive (no weights, no design) & 0.065 & 0.035 & [-0.002, 0.133] & 2942 \\
Weights only (no clustering) & 0.096 & 0.037 & [0.023, 0.170] & 2945 \\
Full design (strata + PSU + weights) & 0.096 & 0.044 & [0.006, 0.187] & 31 \\
Full design + covariates & 0.060 & 0.046 & [-0.033, 0.153] & 31 \\
\bottomrule
\end{tabular}
\\[2pt]\footnotesize \emph{Notes:} The Naive and Weights-only specifications report residual degrees of freedom from the underlying linear regression ($n - p$). The Full design rows report survey degrees of freedom $\sum_h n_h - H$ (sampled PSUs minus strata; see Section~\ref{sec:theory}), which give the appropriate $t$-distribution for design-based inference.
\end{table}

%% file: appendix.tex
% ======================================================================
% Online Appendix
% Loaded via \input{appendix} from main.tex
% ======================================================================

\section{Per-Estimator Influence Function Details}
\label{app:if-details}

This appendix provides the smoothness arguments and IF characterizations for
each modern DiD estimator covered by Proposition~\ref{prop:main}, as well as
the two estimators excluded from the analytical framework.

\subsection{Estimators satisfying Assumption~\ref{ass:smooth}}

\paragraph{Callaway--Sant'Anna (regression).}
$T(F)$ involves population means of outcomes within group-time cells. Sample
means are smooth functionals of $F$. The ATT for cohort $g$ at time $t$ is:
\[
  \ATT(g,t) = \frac{\sum_{i \in G_g} w_i \Delta Y_i}{\sum_{i \in G_g} w_i}
            - \frac{\sum_{i \in C} w_i \Delta Y_i}{\sum_{i \in C} w_i},
\]
where $G_g$ denotes the cohort first treated at $g$ and $C$ the comparison
group. This is a difference of ratios of weighted totals, smooth in the survey-weighted
empirical distribution.

\paragraph{Callaway--Sant'Anna (doubly robust / IPW).}
$T(F)$ additionally involves a propensity score model (smooth in population
moments) and outcome regression (smooth in population moments).
\citet{santannazhao2020} derive the full IF, including nuisance-function
corrections, under their Theorem~3.1. The IF accounts for uncertainty in both
the propensity score and outcome regression, ensuring that the linearization
remainder is $o_p(n^{-1/2})$.

\paragraph{Sun--Abraham.}
$T(F)$ is a linear functional of interaction-weighted regression coefficients,
which are themselves smooth functionals of $F$ via the implicit function
theorem applied to the normal equations. Survey weights enter through WLS.

\paragraph{Imputation DiD (Borusyak--Jaravel--Spiess).}
$T(F)$ involves OLS on untreated observations (smooth), counterfactual
imputation (linear in coefficients), and averaging treatment-minus-imputed
residuals (smooth). The IF follows from Theorem~3 of \citet{borusyak2024},
which provides the full asymptotic representation including the first-stage
estimation uncertainty.

\paragraph{Two-way fixed effects (TWFE).}
Standard OLS regression with unit and time fixed effects. Because the number
of unit fixed effects grows with sample size, the implicit function theorem
applies to the concentrated (within-transformed) estimating equations, where
unit effects have been profiled out via demeaning. Smoothness of the
within-estimator follows directly. Under treatment effect heterogeneity, the
TWFE coefficient is a weighted combination of group-time ATTs with potentially
negative weights \citep{dechaisemartin2020}, but the functional remains smooth
regardless.

\paragraph{Two-stage DiD (Gardner).}
The first stage regresses outcomes on group and period fixed effects (along
with any covariates) using the untreated sample to identify the baseline
two-way structure; the second stage regresses residualized outcomes on
treatment indicators. The IF captures uncertainty from both stages. Smoothness
follows from the GMM representation.

\paragraph{Additional estimators.}
Efficient DiD \citep{chensantannaxie2025} derives the semiparametric efficient
influence function for DiD parameters; smooth by the same arguments as
Callaway--Sant'Anna. Triple difference extends the two-group sandwich to a
triple contrast; the IF follows by the same arguments as standard DiD.
Stacked DiD uses Q-weighted regression on stacked sub-experiments; smooth in
the population moments of each sub-experiment. Continuous DiD involves
B-spline regression coefficients, smooth functionals of $F$.

\paragraph{dCDH \citep{dechaisemartin2020}.}
The dCDH estimator is a smooth weighted average of group-time changes (the
$\mathrm{DID}_\ell$ family of switcher contrasts), where the weights are
group-period sample shares and the contrasts are differences of weighted
means. Smoothness follows under positivity of the comparison-cell denominators
and nondegenerate weights; the IF expands as a sum over group-time cells in
direct analogy to Callaway--Sant'Anna.

\paragraph{Heterogeneous Adoption DiD \citep{dechaisemartin2026had}.}
The HAD estimator targets a Weighted Average Slope (WAS) under heterogeneous
adoption with no fully untreated unit. The continuous-dose path uses a
local-linear nonparametric estimator at the support boundary with a
bias-corrected confidence interval; the mass-point path is a 2SLS
structural-residual contrast. Both are smooth functionals of group-time/cell moments under the
overlap and bandwidth conditions of \citet{dechaisemartin2026had}.

\paragraph{Wooldridge DiD (ETWFE).}
The estimator is a weighted least squares (or IRLS for logit/Poisson)
regression on saturated cohort-by-period interactions
\citep{wooldridge2021}. Smoothness follows from the same implicit-function
argument used for TWFE and Sun--Abraham, applied to the appropriate score
equations.

\paragraph{Staggered Triple Difference.}
The estimator extends the triple-difference IF representation to staggered
group-time contrasts, aggregating cell-level $\mathrm{ATT}(g,t)$ via
sample-share weights as in Callaway--Sant'Anna. Smoothness follows under the
same cell-probability and overlap conditions as the standard triple
difference and Callaway--Sant'Anna.

\subsection{Estimators excluded from the analytical framework}

\paragraph{Synthetic DiD.}
The synthetic control weight selection involves a non-smooth optimization
(constrained least squares on pre-treatment outcomes) that does not fit into
the smooth-functional framework. Survey support uses Rao--Wu rescaled bootstrap
\citep{raowu1988}, which resamples PSUs within strata and re-runs the full
estimator on each bootstrap draw, bypassing the need for an IF.

\paragraph{Triply Robust Panel Estimators (TROP).}
The TROP estimator involves a non-smooth optimization step. Per-observation
treatment effects $\tau_{it}$ are estimated from the data via a first-stage
optimization. The current implementation treats these as fixed across bootstrap
draws, varying only the ATT aggregation weights. This captures the variance of
the aggregation step but not the sampling variability of the first-stage
$\tau_{it}$ estimates, and therefore likely underestimates total variance. This
approximation is a pragmatic choice driven by the computational cost of
re-running the non-smooth optimization for each bootstrap draw; users should
interpret TROP survey standard errors as a lower bound.

% ======================================================================
\section{Replicate Weight Methods}
\label{app:replicate}

Replicate weights provide an alternative to TSL when pre-computed replicate
weight columns accompany the survey data (common for ACS PUMS, CPS, BRFSS,
and RECS; less common for NHANES, where CDC documentation generally
recommends TSL with the masked strata and PSU variables).
Instead of linearizing the estimator, they perturb the weights and observe the
resulting variation in estimates. Replicate weights are mutually exclusive with
strata/PSU/FPC at the design level: the design information is already embedded
in the replicate weight construction.

The general variance formula is:
\begin{equation}
  V_{\text{rep}} = c \sum_{r=1}^{R} s_r
    (\hat{\theta}_r - \hat{\theta}_{\text{center}})^2,
  \label{eq:replicate-general}
\end{equation}
where $\hat{\theta}_r$ is the estimate from replicate $r$,
$\hat{\theta}_{\text{center}}$ is either the full-sample estimate or the mean
of replicate estimates, and $c$ and $s_r$ are method-specific factors:

\noindent For BRR, Fay's BRR, JK1, and SDR, $s_r = 1$ for all replicates and
the method-specific factor $c$ is:

\begin{center}
\begin{tabular}{@{}ll@{}}
\toprule
Method & Factor $c$ \\
\midrule
BRR & $1/R$ \\
Fay's BRR & $1/[R(1-\rho)^2]$ where $\rho$ is the perturbation factor \\
JK1 & $(R-1)/R$ \\
SDR & $4/R$ \\
\bottomrule
\end{tabular}
\end{center}

\noindent For JKn (delete-one-group jackknife), $c = 1$ and
$s_r = (n_h - 1)/n_h$ for replicate $r$ belonging to stratum $h$, giving
$V_{\text{JKn}} = \sum_h \frac{n_h - 1}{n_h} \sum_{r \in h}
(\hat{\theta}_r - \hat{\theta})^2$.

For most IF-based estimators, the replicate estimate is computed by reweighting
the per-unit IF values rather than re-running the estimator --- a
weight-ratio rescaling that is numerically exact to first order and avoids the
cost of $R$ full re-fits. For imputation DiD and two-stage DiD, the analytical influence functions do
account for first-stage estimation uncertainty (as noted in
Appendix~\ref{app:if-details}), and the TSL variance path uses these complete
IFs directly. For the replicate-weight path, a subtlety arises: the
weight-ratio rescaling shortcut computes replicate estimates by reweighting
\emph{stored} IF values that were evaluated at the full-sample nuisance
parameter estimates. For estimators where the first-stage coefficients are
insensitive to small weight perturbations (e.g., CS-reg, where the IFs depend
only on weighted cell means), this first-order approximation is accurate. For
imputation DiD and two-stage DiD, the first-stage regression coefficients shift
materially with the replicate weights, making the stored-IF approximation less
accurate. The replicate-weight path therefore requires full re-estimation with
each replicate's weights, which produces exact (not linearized) replicate
estimates at the cost of $R$ additional fits.

%% file: references.bib
@article{binder1983,
  author  = {Binder, David A.},
  title   = {On the Variances of Asymptotically Normal Estimators from Complex Surveys},
  journal = {International Statistical Review},
  volume  = {51},
  number  = {3},
  pages   = {279--292},
  year    = {1983}
}

@article{demnati2004,
  author  = {Demnati, Abdellatif and Rao, J. N. K.},
  title   = {Linearization Variance Estimators for Survey Data},
  journal = {Survey Methodology},
  volume  = {30},
  number  = {1},
  pages   = {17--26},
  year    = {2004}
}

@article{lumley2004,
  author  = {Lumley, Thomas},
  title   = {Analysis of Complex Survey Samples},
  journal = {Journal of Statistical Software},
  volume  = {9},
  number  = {8},
  pages   = {1--19},
  year    = {2004}
}

@article{raowu1988,
  author  = {Rao, J. N. K. and Wu, C. F. Jeff},
  title   = {Resampling Inference with Complex Survey Data},
  journal = {Journal of the American Statistical Association},
  volume  = {83},
  number  = {401},
  pages   = {231--241},
  year    = {1988}
}

@article{atheyimbens2022,
  author  = {Athey, Susan and Imbens, Guido W.},
  title   = {Design-Based Analysis in Difference-in-Differences Settings with Staggered Adoption},
  journal = {Journal of Econometrics},
  volume  = {226},
  number  = {1},
  pages   = {62--79},
  year    = {2022}
}

@article{borusyak2024,
  author  = {Borusyak, Kirill and Jaravel, Xavier and Spiess, Jann},
  title   = {Revisiting Event-Study Designs: Robust and Efficient Estimation},
  journal = {Review of Economic Studies},
  volume  = {91},
  number  = {6},
  pages   = {3253--3285},
  year    = {2024}
}

@article{callawaysantanna2021,
  author  = {Callaway, Brantly and Sant'Anna, Pedro H. C.},
  title   = {Difference-in-Differences with Multiple Time Periods},
  journal = {Journal of Econometrics},
  volume  = {225},
  number  = {2},
  pages   = {200--230},
  year    = {2021}
}

@unpublished{callaway2024continuous,
  author = {Callaway, Brantly and Goodman-Bacon, Andrew and Sant'Anna, Pedro H. C.},
  title  = {Difference-in-Differences with a Continuous Treatment},
  note   = {NBER Working Paper 32117},
  year   = {2024}
}

@article{dechaisemartin2020,
  author  = {{de Chaisemartin}, Cl\'{e}ment and D'Haultf{\oe}uille, Xavier},
  title   = {Two-Way Fixed Effects Estimators with Heterogeneous Treatment Effects},
  journal = {American Economic Review},
  volume  = {110},
  number  = {9},
  pages   = {2964--2996},
  year    = {2020}
}

@unpublished{dechaisemartin2026had,
  author = {{de Chaisemartin}, Cl\'{e}ment and Ciccia, Diego and D'Haultf{\oe}uille, Xavier and Knau, Felix},
  title  = {Difference-in-Differences Estimators When No Unit Remains Untreated},
  note   = {arXiv:2405.04465v6},
  year   = {2026}
}

@unpublished{gardner2022,
  author = {Gardner, John},
  title  = {Two-Stage Differences in Differences},
  note   = {Working Paper},
  year   = {2022}
}

@article{roth2023,
  author  = {Roth, Jonathan and Sant'Anna, Pedro H. C. and Bilinski, Alyssa and Poe, John},
  title   = {What's Trending in Difference-in-Differences? {A} Synthesis of the Recent Econometrics Literature},
  journal = {Journal of Econometrics},
  volume  = {235},
  number  = {2},
  pages   = {2218--2244},
  year    = {2023}
}

@article{santannazhao2020,
  author  = {Sant'Anna, Pedro H. C. and Zhao, Jun},
  title   = {Doubly Robust Difference-in-Differences Estimators},
  journal = {Journal of Econometrics},
  volume  = {219},
  number  = {1},
  pages   = {101--122},
  year    = {2020}
}

@article{sunabraham2021,
  author  = {Sun, Liyang and Abraham, Sarah},
  title   = {Estimating Dynamic Treatment Effects in Event Studies with Heterogeneous Treatment Effects},
  journal = {Journal of Econometrics},
  volume  = {225},
  number  = {2},
  pages   = {175--199},
  year    = {2021}
}

@article{dugoff2014,
  author  = {DuGoff, Eva H. and Schuler, Megan and Stuart, Elizabeth A.},
  title   = {Generalizing Observational Study Results: Applying Propensity Score Methods to Complex Surveys},
  journal = {Health Services Research},
  volume  = {49},
  number  = {1},
  pages   = {284--303},
  year    = {2014}
}

@article{solon2015,
  author  = {Solon, Gary and Haider, Steven J. and Wooldridge, Jeffrey M.},
  title   = {What Are We Weighting For?},
  journal = {Journal of Human Resources},
  volume  = {50},
  number  = {2},
  pages   = {301--316},
  year    = {2015}
}

@article{ye2025,
  author  = {Ye, Kerry and Bilinski, Alyssa and Lee, Youjin},
  title   = {Difference-in-differences analysis with repeated cross-sectional survey data},
  journal = {Health Services \& Outcomes Research Methodology},
  year    = {2025},
  doi     = {10.1007/s10742-025-00364-7}
}

@article{zeng2025,
  author  = {Zeng, Yukang and Li, Fan and Tong, Guangyu},
  title   = {Moving toward Best Practice when Using Propensity Score Weighting in Survey Observational Studies},
  journal = {Statistics in Medicine},
  year    = {2025},
  doi     = {10.1002/sim.70555}
}

@article{antwi2013,
  author  = {Antwi, Yaa Akosa and Moriya, Asako S. and Simon, Kosali},
  title   = {Effects of Federal Policy to Insure Young Adults: Evidence from the 2010 {Affordable Care Act}'s Dependent-Coverage Mandate},
  journal = {American Economic Journal: Economic Policy},
  volume  = {5},
  number  = {4},
  pages   = {1--28},
  year    = {2013}
}

@article{sommers2012,
  author  = {Sommers, Benjamin D. and Kronick, Richard},
  title   = {The {Affordable Care Act} and Insurance Coverage for Young Adults},
  journal = {JAMA},
  volume  = {307},
  number  = {9},
  pages   = {913--914},
  year    = {2012}
}

@misc{nchs2018,
  author       = {{National Center for Health Statistics}},
  title        = {National Health and Nutrition Examination Survey: Analytic Guidelines, 2011--2016},
  howpublished = {Hyattsville, MD: U.S. Department of Health and Human Services},
  year         = {2018},
  url          = {https://wwwn.cdc.gov/nchs/data/nhanes/analyticguidelines/11-16-analytic-guidelines.pdf}
}

@unpublished{chensantannaxie2025,
  author = {Chen, Xiaohong and Sant'Anna, Pedro H. C. and Xie, Haitian},
  title  = {Efficient Difference-in-Differences and Event Study Estimators},
  note   = {arXiv:2506.17729},
  year   = {2025}
}

@article{wooldridge2021,
  author  = {Wooldridge, Jeffrey M.},
  title   = {Two-Way Fixed Effects, the Two-Way {Mundlak} Regression, and Difference-in-Differences Estimators},
  journal = {Empirical Economics},
  volume  = {69},
  number  = {5},
  pages   = {2545--2587},
  year    = {2025},
  note    = {First circulated as SSRN WP 3906345, 2021}
}

@unpublished{athey2025trop,
  author = {Athey, Susan and Imbens, Guido and Qu, Zhaonan and Viviano, Davide},
  title  = {Triply Robust Panel Estimators},
  note   = {arXiv:2508.21536},
  year   = {2025}
}

@article{cengiz2019,
  author  = {Cengiz, Doruk and Dube, Arindrajit and Lindner, Attila and Zipperer, Ben},
  title   = {The Effect of Minimum Wages on Low-Wage Jobs},
  journal = {Quarterly Journal of Economics},
  volume  = {134},
  number  = {3},
  pages   = {1405--1454},
  year    = {2019}
}

@article{bellmccaffrey2002,
  author  = {Bell, Robert M. and McCaffrey, Daniel F.},
  title   = {Bias Reduction in Standard Errors for Linear Regression with Multi-Stage Samples},
  journal = {Survey Methodology},
  volume  = {28},
  number  = {2},
  pages   = {169--181},
  year    = {2002}
}

@article{faygraubard2001,
  author  = {Fay, Michael P. and Graubard, Barry I.},
  title   = {Small-Sample Adjustments for {Wald}-Type Tests Using Sandwich Estimators},
  journal = {Biometrics},
  volume  = {57},
  number  = {4},
  pages   = {1198--1206},
  year    = {2001}
}

@misc{diffdiff2026,
  author    = {Gerber, Isaac},
  license   = {MIT},
  month     = apr,
  title     = {{diff-diff: Difference-in-Differences Causal Inference for Python}},
  publisher = {Zenodo},
  doi       = {10.5281/zenodo.19803705},
  url       = {https://doi.org/10.5281/zenodo.19803705},
  version   = {3.3.2},
  year      = {2026}
}

@misc{gerberreplication2026,
  author    = {Gerber, Isaac},
  license   = {MIT},
  title     = {Replication code: Design-Based Variance Estimation for Modern
               Heterogeneity-Robust Difference-in-Differences Estimators},
  publisher = {Zenodo},
  doi       = {10.5281/zenodo.20097361},
  url       = {https://doi.org/10.5281/zenodo.20097361},
  version   = {v1.0},
  year      = {2026}
}
